\documentclass[oneside,american]{amsart}
\usepackage{mathptmx}

\usepackage[T1]{fontenc}
\usepackage[latin9]{inputenc}
\usepackage{geometry}
\usepackage{amsthm}
\usepackage{amssymb}
\usepackage{graphicx}

\makeatletter
\numberwithin{equation}{section}
\numberwithin{figure}{section}
\theoremstyle{plain}
\newtheorem{thm}{\protect\theoremname}
\theoremstyle{plain}
\newtheorem{lem}[thm]{\protect\lemmaname}
\theoremstyle{remark}
\newtheorem{rem}[thm]{\protect\remarkname}

\makeatother

\usepackage{babel}
\providecommand{\lemmaname}{Lemma}
\providecommand{\remarkname}{Remark}
\providecommand{\theoremname}{Theorem}

\begin{document}
\title{On representations of the Helmholtz Green's function}
\author{Gregory Beylkin}
\address{Department of Applied Mathematics, University of Colorado at Boulder,
UCB 526, Boulder, CO 80309-0526}
\begin{abstract}
We consider the free space Helmholtz Green's function and split it
into the sum of oscillatory and non-oscillatory (singular) components.
The goal is to separate the impact of the singularity of the real
part at the origin from the oscillatory behavior controlled by the
wave number $k$. The oscillatory component can be chosen to have
any finite number of continuous derivatives at the origin and can
be applied to a function in the Fourier space in $\mathcal{O}\left(k^{d}\log k\right)$
operations. The non-oscillatory component has a multiresolution representation
via a linear combination of Gaussians and is applied efficiently in
space.

Since the Helmholtz Green's function can be viewed as a point source,
this partitioning can be interpreted as a splitting into propagating
and evanescent components. We show that the non-oscillatory component
is significant only in the vicinity of the source at distances $\mathcal{O}\left(c_{1}k^{-1}+c_{2}k^{-1}\log_{10}k\right)$,
for some constants $c_{1}$, $c_{2}$, whereas the propagating component
can be observed at large distances.
\end{abstract}

\maketitle

\section{Introduction}

In this paper we consider the free space Helmholtz Green's function
given by
\begin{equation}
G\left(r\right)=\left\{ \begin{array}{rc}
{\displaystyle \frac{1}{4\pi}\frac{e^{ikr}}{r}=\frac{1}{4\pi}\frac{\cos\left(kr\right)}{r}+\frac{i}{4\pi}\frac{\sin\left(kr\right)}{r}} & \mbox{in dimension }d=3,\\
\\
{\displaystyle \frac{i}{4}H_{0}^{(1)}\left(kr\right)=-\frac{1}{4}Y_{0}\left(kr\right)+\frac{i}{4}J_{0}\left(kr\right)} & \mbox{in dimension }d=2,
\end{array}\right.\label{eq:  FreeSpaceGF}
\end{equation}
where $H_{0}^{(1)}$ is the Hankel function of the first kind, $J_{0}$
and $Y_{0}$ are the Bessel functions of the first and second kind,
$r=\left\Vert \mathbf{x}\right\Vert =\left(\sum_{j=1}^{d}x_{j}^{2}\right)^{1/2}$
denotes the Euclidean norm of the vector $\mathbf{x}$ and $k>0$.
We separate $G$ into the sum of oscillatory and non-oscillatory (singular)
components. The oscillatory component can be chosen to have any finite
number of continuous derivatives at $r=0$. As far as we know, previous
approaches to split $G$ in this manner did not allow to choose the
number of smooth derivatives at $r=0$. As in \cite{BE-KU-MO:2009,BE-KU-MO:2008},
the oscillatory component can be applied to a function in the Fourier
space in $\mathcal{O}\left(k^{d}\log k\right)$ operations. The non-oscillatory
component has a multiresolution representation via a linear combination
of Gaussians and is applied efficiently in space.

Our approach is a modification of that in \cite{BE-KU-MO:2009,BE-KU-MO:2008}
leading to explicit formulas. The goal is to separate the impact of
the singularity of the real part of (\ref{eq:  FreeSpaceGF}) at the
origin from the oscillatory behavior controlled by the wave number
$k$. Specifically, we want the number of derivatives at the origin
of the oscillatory component to be user selected and the non-oscillatory
component to have a multiresolution representation via a linear combination
of Gaussians. For non-oscillatory kernels integral representations
involving Gaussians lead to efficient multiresolution approximations
(see e.g. \cite{H-F-Y-G-B:2004,BEY-MOH:2005,B-C-F-H:2007,BE-MO-PE:2008,BEY-MON:2010,B-F-H-K-M:2012,H-B-B-C-F-F-G-etc:2016,A-H-S-T-B:2020}),
i.e. when approaching a singularity the domain of integration automatically
shrinks leading to fast algorithms for application of such kernels.
We want a similar representation of the non-oscillatory component
of the Helmholtz Green's function (\ref{eq:  FreeSpaceGF}).

Since $G$ in (\ref{eq:  FreeSpaceGF}) can be interpreted as a point
source, a physical interpretation of splitting it into oscillatory
and non-oscillatory components may be viewed as a splitting into propagating
and evanescent components. Indeed, we show that the non-oscillatory
component is significant only in the vicinity of the source at distances
$\mathcal{O}\left(c_{1}k^{-1}+c_{2}k^{-1}\log_{10}k\right)$, for
constants $c_{1}$, $c_{2}$, whereas the propagating component can
be observed at large distances.

\section{Preliminaries}

\subsection{Green's functions}

The free space Green's function (\ref{eq:  FreeSpaceGF}) of the Helmholtz
equation satisfies
\begin{equation}
\Delta G\left(\mathbf{x}\right)+k^{2}G\left(\mathbf{x}\right)=-\delta\left(\mathbf{x}\right),\label{eq:  Fundamental Solution}
\end{equation}
and, on taking the Fourier transform of (\ref{eq:  Fundamental Solution}),
we obtain 
\begin{equation}
\widehat{G}\left(\left\Vert \mathbf{p}\right\Vert \right)=\frac{1}{\left\Vert \mathbf{p}\right\Vert ^{2}-k^{2}},\label{eq:  G hat}
\end{equation}
where $\mathbf{p}\in\mathbb{R}^{d}$, $d=2,3$, We use the Fourier
transform defined as 
\begin{equation}
\widehat{f}\left(\mathbf{p}\right)=\int_{\mathbb{R}^{d}}f\left(\mathbf{x}\right)e^{-i\mathbf{x}\cdot\mathbf{p}}d\mathbf{x}\label{eq:  Forward FT}
\end{equation}
and its inverse as
\begin{equation}
f\left(\mathbf{x}\right)=\frac{1}{\left(2\pi\right)^{d}}\int_{\mathbb{R}^{d}}\widehat{f}\left(\mathbf{p}\right)e^{i\mathbf{x}\cdot\mathbf{p}}d\mathbf{p}.\label{eq:  Inverse FT}
\end{equation}
The inverse Fourier transform of $\widehat{G}$ is a singular integral
and we use regularization (see \cite{BE-KU-MO:2009})
\begin{equation}
G(\mathbf{x})=\lim_{\lambda\rightarrow0^{+}}\frac{1}{(2\pi)^{d}}\int_{\mathbb{R}^{d}}\frac{e^{i\mathbf{x}\cdot\mathbf{p}}}{\left\Vert \mathbf{p}\right\Vert ^{2}-\left(k+i\lambda\right)^{2}}d\mathbf{p},\label{eq:  G Limiting}
\end{equation}
which yields the outgoing Green's functions (\ref{eq:  FreeSpaceGF})
satisfying the Sommerfeld radiation condition
\begin{equation}
\lim_{r\rightarrow\infty}r{}^{\frac{d-1}{2}}\left(\frac{\partial G}{\partial r}-ikG\right)=0.\label{eq:  Sommerfeld}
\end{equation}
Following \cite{BE-KU-MO:2009}, we define 
\[
\hat{G}_{\lambda}\left(p\right)=\frac{1}{p^{2}-\left(k+i\lambda\right)^{2}}=\frac{1}{2p}\left(\frac{1}{p-k-i\lambda}+\frac{1}{p+k+i\lambda}\right)
\]
and separate its real and imaginary parts

\noindent 
\begin{equation}
\mathcal{R}e\left(\hat{G}_{\lambda}\left(p\right)\right)=\frac{1}{2p}\left(\frac{p-k}{\left(p-k\right){}^{2}+\lambda^{2}}+\frac{p+k}{\left(p+k\right){}^{2}+\lambda^{2}}\right)\label{eq:  Real G hat eps}
\end{equation}
and
\begin{equation}
\mathcal{I}m\left(\hat{G}_{\lambda}\left(p\right)\right)=\frac{1}{2p}\left(\frac{\lambda}{\left(p-k\right){}^{2}+\lambda^{2}}-\frac{\lambda}{\left(p+k\right){}^{2}+\lambda^{2}}\right).\label{eq:  Imag G hat eps}
\end{equation}

\noindent We observe that the limit 
\begin{equation}
\lim_{\lambda\rightarrow0^{+}}\mathcal{I}m\left(\hat{G}_{\lambda}\left(p\right)\right)=\frac{\pi}{2p}\left(\delta\left(p-k\right)-\delta\left(p+k\right)\right)\label{eq:  Imag G hat delta function}
\end{equation}
is a generalized function (see e.g. \cite[Chapter III, section 1.3]{GEL-SHI:1964})
corresponding to integration over the sphere in the Fourier domain.
Using e.g. \cite[Section 4.1]{GRAFAK:2004}, we have 
\[
\lim_{\lambda\rightarrow0^{+}}\int_{0}^{\infty}\frac{1}{2p}\left(\frac{p-k}{\left(p-k\right)^{2}+\lambda^{2}}+\frac{p+k}{\left(p+k\right){}^{2}+\lambda^{2}}\right)d\rho=\mbox{p.v.}\int_{0}^{\infty}\frac{1}{p^{2}-k^{2}}d\rho,
\]

\noindent so that
\begin{equation}
\mathcal{R}e\left(G\left(\mathbf{x}\right)\right)=\frac{1}{(2\pi)^{d}}\mbox{p.v.}\int_{\mathbb{R}^{d}}\frac{e^{i\mathbf{x}\cdot\mathbf{p}}}{\left\Vert \mathbf{p}\right\Vert ^{2}-k^{2}}d\mathbf{p},\label{eq:  reGpv}
\end{equation}
where the principal value is considered about $\left\Vert \mathbf{p}\right\Vert =k$.

\section{Splitting of the Green's function in the Fourier domain}

We start with
\begin{lem}
\label{lem:key lemma}For $n\ge1$ and $p\ne k$ we have
\begin{equation}
\frac{1}{p^{2}-k^{2}}-\widehat{g}_{n}\left(p,k\right)=\frac{1}{p^{2}-k^{2}}\frac{\left(2k^{2}\right)^{n}}{\left(p^{2}+k^{2}\right)^{n}}=\mathcal{O}\left(\frac{1}{p^{2n+2}}\right),\label{eq:main approx}
\end{equation}
where
\begin{equation}
\widehat{g}_{n}\left(p,k\right)=\sum_{j=0}^{n-1}\frac{\left(2k^{2}\right)^{j}}{\left(p^{2}+k^{2}\right)^{j+1}}.\label{eq:g_n}
\end{equation}
\end{lem}

\begin{proof}
We have
\begin{eqnarray*}
\widehat{g}_{n}\left(p,k\right) & = & \frac{1}{p^{2}+k^{2}}\sum_{j=0}^{n-1}\frac{\left(2k^{2}\right)^{j}}{\left(p^{2}+k^{2}\right)^{j}}\\
 & = & \frac{1}{p^{2}+k^{2}}\left(1-\frac{\left(2k^{2}\right)^{n}}{\left(p^{2}+k^{2}\right)^{n}}\right)\left(1-\frac{2k^{2}}{p^{2}+k^{2}}\right)^{-1}\\
 & = & \frac{1}{p^{2}-k^{2}}\left(1-\frac{\left(2k^{2}\right)^{n}}{\left(p^{2}+k^{2}\right)^{n}}\right)\\
 & = & \frac{1}{p^{2}-k^{2}}-\frac{1}{p^{2}-k^{2}}\frac{\left(2k^{2}\right)^{n}}{\left(p^{2}+k^{2}\right)^{n}}
\end{eqnarray*}
and arrive at (\ref{eq:main approx}) as an algebraic identity for
$p\ne k$.
\end{proof}
Using Lemma~\ref{lem:key lemma}, we obtain the splitting of (\ref{eq:  G hat})
in the Fourier domain as 
\begin{equation}
\widehat{G}\left(\left\Vert \mathbf{p}\right\Vert \right)=\widehat{g}_{n}\left(\left\Vert \mathbf{p}\right\Vert ,k\right)+\widehat{g}_{n,oscill}\left(\left\Vert \mathbf{p}\right\Vert ,k\right),\label{eq:splitting}
\end{equation}
where 
\begin{equation}
\widehat{g}_{n,oscill}\left(\left\Vert \mathbf{p}\right\Vert ,k\right)=\frac{1}{\left\Vert \mathbf{p}\right\Vert ^{2}-k^{2}}\frac{\left(2k^{2}\right)^{n}}{\left(\left\Vert \mathbf{p}\right\Vert ^{2}+k^{2}\right)^{n}}\label{eq:g hat oscill}
\end{equation}
and 
\begin{equation}
\widehat{g}_{n}\left(\left\Vert \mathbf{p}\right\Vert ,k\right)=\sum_{j=0}^{n-1}\frac{\left(2k^{2}\right)^{j}}{\left(\left\Vert \mathbf{p}\right\Vert ^{2}+k^{2}\right)^{j+1}}.\label{eq: non-oscillatory}
\end{equation}
The rate of decay of $\widehat{g}_{n,oscill}$ in the Fourier domain
for $\left\Vert \mathbf{p}\right\Vert >k$ is $\mathcal{O}\left(\left\Vert \mathbf{p}\right\Vert ^{-2n-2}\right)$
so that the volume of its significant support is proportional to $k^{d}$.
Following \cite{BE-KU-MO:2009,BE-KU-MO:2008}, we have
\begin{eqnarray*}
\widehat{G}\left(\left\Vert \mathbf{p}\right\Vert \right)=\frac{1}{2\left\Vert \mathbf{p}\right\Vert }\left(\frac{1}{\left\Vert \mathbf{p}\right\Vert -k}+\frac{1}{\left\Vert \mathbf{p}\right\Vert +k}\right) & =\\
\frac{1}{2\left\Vert \mathbf{p}\right\Vert }\int_{0}^{\infty}\left[\left(\left\Vert \mathbf{p}\right\Vert -k\right)e^{-s\left(\left\Vert \mathbf{p}\right\Vert -k\right)^{2}}+\left(\left\Vert \mathbf{p}\right\Vert +k\right)e^{-s\left(\left\Vert \mathbf{p}\right\Vert +k\right)^{2}}\right]ds
\end{eqnarray*}
yielding 
\[
\widehat{g}_{n,oscill}\left(\left\Vert \mathbf{p}\right\Vert ,k\right)=\frac{2^{n-1}k^{2n}}{\left(\left\Vert \mathbf{p}\right\Vert ^{2}+k^{2}\right)^{n}}\frac{1}{\left\Vert \mathbf{p}\right\Vert }\int_{0}^{\infty}\left[\left(\left\Vert \mathbf{p}\right\Vert -k\right)e^{-s\left(\left\Vert \mathbf{p}\right\Vert -k\right)^{2}}+\left(\left\Vert \mathbf{p}\right\Vert +k\right)e^{-s\left(\left\Vert \mathbf{p}\right\Vert +k\right)^{2}}\right]ds,
\]
or
\begin{equation}
\widehat{g}_{n,oscill}\left(\left\Vert \mathbf{p}\right\Vert ,k\right)=\frac{2^{n-1}k^{2n}}{\left(\left\Vert \mathbf{p}\right\Vert ^{2}+k^{2}\right)^{n}}\frac{1}{\left\Vert \mathbf{p}\right\Vert }\int_{-\infty}^{\infty}\left[\left(\left\Vert \mathbf{p}\right\Vert -k\right)e^{-e^{t}\left(\left\Vert \mathbf{p}\right\Vert -k\right)^{2}}+\left(\left\Vert \mathbf{p}\right\Vert +k\right)e^{-e^{t}\left(\left\Vert \mathbf{p}\right\Vert +k\right)^{2}}\right]e^{t}dt,\label{eq:g_=00007Bn,oscill=00007D}
\end{equation}
where $\mathbf{p}\in\mathbb{R}^{d}$, $d=2,3$. The integral in (\ref{eq:g_=00007Bn,oscill=00007D})
is a multiresolution representation of $\widehat{g}_{n,oscill}$ centered
at the singularity $\left\Vert \mathbf{p}\right\Vert =k$. This representation
allows us to implement the principal value limit (see e.g. (\ref{eq:  reGpv}))
by simply ignoring fine scales since, at some point, their contribution
is negligible. Effectively it amounts to replacing the upper limit
in the integral in (\ref{eq:g_=00007Bn,oscill=00007D}) by a carefully
chosen finite value. Discretizing (\ref{eq:g_=00007Bn,oscill=00007D})
leads to an approximation of $\widehat{g}_{n,oscill}$ via a linear
combination of smooth (rotationally invariant) kernels similar to
that obtained in \cite{BE-KU-MO:2009}. We refer to \cite{BE-KU-MO:2009}
for the details of applying $\widehat{g}_{n,oscill}$ to a function
via an algorithm of complexity $\mathcal{O}\left(k^{d}\log k\right)$.

\section{Spatial representations in $\mathbb{R}^{3}$}

While the oscillatory component is applied efficiently in the Fourier
domain due to its rapid decay, the non-oscillatory component decays
slowly in the Fourier domain but its application is efficient in space.
We start by computing spatial representations in $\mathbb{R}^{3}$
since these are different in dimensions $d=3$ and $d=2$. 
\begin{lem}
The inverse Fourier transform of the non-oscillatory component (\ref{eq: non-oscillatory})
is given by
\begin{equation}
g_{n}\left(r,k\right)=\frac{1}{4\pi}\frac{e^{-kr}}{r}\left(1+\sum_{j=1}^{n-1}\frac{1}{2^{j-1}j!}\sum_{m=0}^{j-1}\frac{\left(2j-m-2\right)!2^{m}}{m!\left(j-m-1\right)!}\left(kr\right)^{m+1}\right)\label{eq:general formula for g_n(r)}
\end{equation}
Computing $g_{n}$ for $n=1,2,\dots$ we obtain
\begin{eqnarray}
g_{1}\left(r,k\right) & = & \frac{1}{4\pi}\frac{e^{-kr}}{r}\nonumber \\
g_{2}\left(r,k\right) & = & \frac{1}{4\pi}\frac{e^{-kr}}{r}\left(1+kr\right)\nonumber \\
g_{3}\left(r,k\right) & = & \frac{1}{4\pi}\frac{e^{-kr}}{r}\left(1+\frac{3}{2}kr+\frac{1}{2}\left(kr\right)^{2}\right)\nonumber \\
g_{4}\left(r,k\right) & = & \frac{1}{4\pi}\frac{e^{-kr}}{r}\left(1+2kr+\left(kr\right)^{2}+\frac{1}{6}\left(kr\right)^{3}\right)\label{eq:explicit reps for g_n}\\
g_{5}\left(r,k\right) & = & \frac{1}{4\pi}\frac{e^{-kr}}{r}\left(1+\frac{21}{8}kr+\frac{13}{8}\left(kr\right)^{2}+\frac{5}{12}\left(kr\right)^{3}+\frac{1}{24}\left(kr\right)^{4}\right)\nonumber \\
 & \dots & ,\nonumber 
\end{eqnarray}
where $r=\left\Vert \mathbf{x}\right\Vert $.
\end{lem}

\begin{proof}
Computing the inverse Fourier transform of rotationally invariant
function (\ref{eq: non-oscillatory}) in dimension $d=3$, we have 

\begin{equation}
g_{n}\left(r,k\right)=\frac{1}{\left(2\pi\right)^{3}}\int_{\mathbb{R}^{3}}\widehat{g}_{n}\left(\left\Vert \mathbf{p}\right\Vert ,k\right)e^{i\mathbf{p}\cdot\mathbf{r}}d\mathbf{p}=\frac{1}{2\pi^{2}r}\sum_{j=0}^{n-1}\left(2k^{2}\right)^{j}\int_{0}^{\infty}\frac{p\sin\left(pr\right)}{\left(p^{2}+k^{2}\right)^{j+1}}dp,\label{eq: explicit g_n(r) in 3D}
\end{equation}
where the last integral is available in \cite[Formula 3.737.2]{GRA-RYZ:2007}
leading to (\ref{eq:general formula for g_n(r)}).
\end{proof}
\begin{rem}
\label{rem:We-want-to}We want to estimate the significant support
of $g_{n}$, $g_{n}\left(r,k\right)\ge\epsilon$ as a function of
$k$. Ignoring constant factors, we observe from (\ref{eq:general formula for g_n(r)})
that the term in (\ref{eq:general formula for g_n(r)}) $e^{-kr}r^{n-2}k^{n-1}$
decays slower than other terms. To estimate its significant support,
we consider $e^{-kr}r^{n-2}k^{n-1}\ge\epsilon$, so that
\begin{equation}
r\le k^{-1}\log_{10}\left(\epsilon^{-1}\right)+k^{-1}\left(n-2\right)\log_{10}r+k^{-1}\left(n-1\right)\log_{10}k.\label{eq:inequality}
\end{equation}
Although $r$ appears on both sides of this inequality, the factor
$\log_{10}r$ is negative for $r\le1$ and, since we are interested
in distances $\mathcal{O}\left(k^{-1}\log_{10}k\right)$, the second
term in (\ref{eq:inequality}) can be dropped. For a given $\epsilon$,
as $k$ becomes large, $g_{n}\left(r,k\right)$ is greater than $\epsilon$
within a ball of radius of $\mathcal{O}\left(c_{1}k^{-1}+c_{2}k^{-1}\log_{10}k\right)$.
We illustrate this relation in Figure~\ref{fig:Plots-of-evanescent}
observing that, for a fixed $\epsilon$, $\log r$ is essentially
proportional to $-\log k$. 
\begin{figure}
\begin{centering}
\includegraphics[scale=0.45]{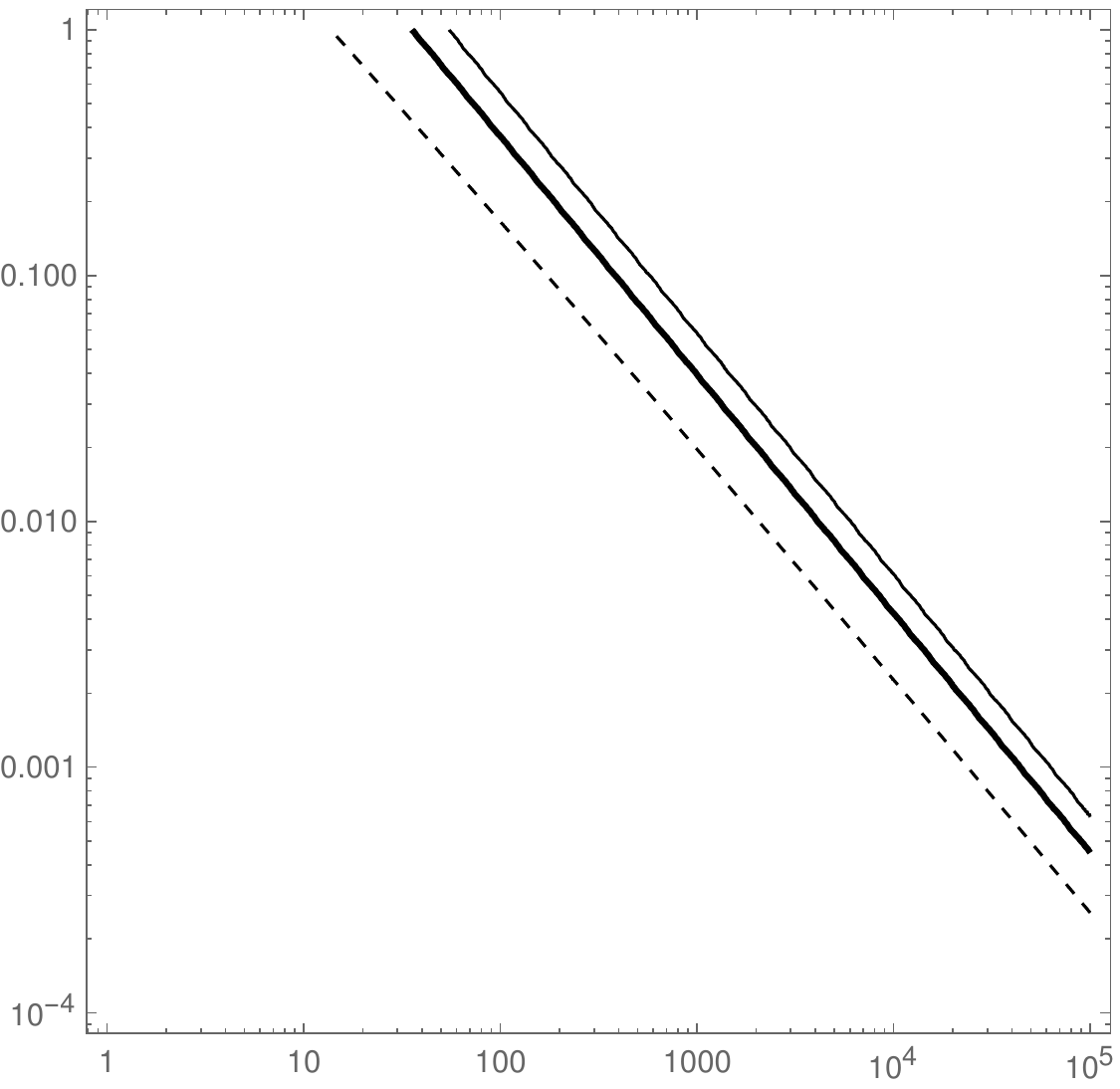}
\par\end{centering}
\caption{\label{fig:Plots-of-evanescent}Log-log contour plot of the non-oscillatory
component $g_{8}\left(r\left(k\right),k\right)$ in (\ref{eq:general formula for g_n(r)})
for different values of $\epsilon$ and $1\le k\le10^{5}$, where thin line corresponds to $\epsilon=10^{-16}$, thick line
to $\epsilon=10^{-9}$ and dashed line to $10^{-2}$.}
\end{figure}
\end{rem}

Next we consider the difference between the real part of $G$ in (\ref{eq:  FreeSpaceGF})
and $g_{n}$,
\begin{equation}
q_{n}\left(r,k\right)=\mathcal{R}e\left(g_{n,oscill}\right)\left(r,k\right)=\frac{1}{4\pi}\frac{\cos kr}{r}-g_{n}\left(r,k\right),\label{eq:difference q_n}
\end{equation}
and examine its behavior at $r=0$. 
\begin{lem}
\label{lem:The-Taylor-expansion} The difference 
\begin{equation}
\mathcal{R}e\left(g_{n,oscill}\right)\left(\mathbf{r},k\right)=\frac{1}{4\pi}\frac{\cos k\left\Vert \mathbf{r}\right\Vert }{\left\Vert \mathbf{r}\right\Vert }-g_{n}\left(\left\Vert \mathbf{r}\right\Vert ,k\right),\label{eq:difference cos/r and g_n-1}
\end{equation}
has continuous partial derivatives at zero up to order $2n-2$. The
Taylor expansion of $q_{n}$ at $r=\left\Vert \mathbf{r}\right\Vert =0$
yields
\begin{eqnarray}
q_{1}\left(r,k\right) & = & \frac{k}{4\pi}\left(1-kr+\mathcal{O}\left(\left(kr\right)^{2}\right)\right)\nonumber \\
q_{2}\left(r,k\right) & = & \frac{k}{4\pi}\left(-\frac{\left(kr\right)^{2}}{3}+\frac{\left(kr\right)^{3}}{6}+\mathcal{O}\left(\left(kr\right)^{4}\right)\right)\nonumber \\
q_{3}\left(r,k\right) & = & \frac{k}{4\pi}\left(-\frac{1}{2}-\frac{\left(kr\right)^{2}}{12}+\frac{7\left(kr\right)^{4}}{240}-\frac{\left(kr\right)^{5}}{90}+\mathcal{O}\left(\left(kr\right)^{6}\right)\right)\nonumber \\
q_{4}\left(r,k\right) & = & \frac{k}{4\pi}\left(-1+\frac{\left(kr\right)^{4}}{120}-\frac{\left(kr\right)^{6}}{840}+\frac{\left(kr\right)^{7}}{2520}+\mathcal{O}\left(\left(kr\right)^{8}\right)\right)\label{eq:expansion at zero 3D}\\
q_{5}\left(r,k\right) & = & \frac{k}{4\pi}\left(\right.-\frac{13}{8}+\frac{\left(kr\right)^{2}}{16}+\frac{\left(kr\right)^{4}}{320}-\frac{\left(kr\right)^{6}}{40320}+\nonumber \\
 &  & \frac{83\left(kr\right)^{8}}{2903040}-\frac{\left(kr\right)^{9}}{113400}+\mathcal{O}\left(\left(kr\right)^{10}\right)\left.\right)\nonumber \\
 & \dots & .\nonumber 
\end{eqnarray}
\end{lem}

\begin{proof}
The function $q_{n}$ in (\ref{eq:difference q_n}) is the inverse
Fourier transform of rotationally invariant function (\ref{eq:g hat oscill}),
\begin{equation}
\mathcal{R}e\left(g_{n,oscill}\right)\left(r,k\right)=\frac{1}{2\pi^{2}r}\mbox{p.v.}\int_{0}^{\infty}\frac{p}{p^{2}-k^{2}}\frac{\left(2k^{2}\right)^{n}}{\left(p^{2}+k^{2}\right)^{n}}\sin\left(pr\right)dp.\label{eq: g_n,oscill}
\end{equation}
Formally $\mathcal{R}e\left(g_{n,oscill}\right)$ in (\ref{eq: g_n,oscill})
is an even function of $r$ and, as long as the necessary derivatives
exist, only even powers can appear in its Taylor expansion. Taking
$2n-2$ derivatives of $g_{n,oscill}\left(r,k\right)$ with respect
to $r$ correspond to multiplying the integrand in (\ref{eq: g_n,oscill})
by powers $p^{j}$, $0<j\le2n-2$ which changes the rate of decay
of the integrand from $\mathcal{O}\left(p^{-2n-1}\right)$ to as low
as $\mathcal{O}\left(p^{3}\right)$ so that the integrals for these
derivatives exist. Replacing $r$ by $\left\Vert \mathbf{r}\right\Vert $,
we observe that continuous partial derivatives at zero exist up to
the order $2n-2$. As we see in (\ref{eq:expansion at zero 3D}),
the next term in the expansion does not yield a continuous derivative
of (\ref{eq:difference cos/r and g_n-1}) at zero. Several examples
of expansions of (\ref{eq:difference q_n}) using (\ref{eq:general formula for g_n(r)})
are presented in (\ref{eq:expansion at zero 3D}) .
\end{proof}
In Figure~\ref{fig:Real-and-imaginary} we plot $g_{4,oscill}$ to
illustrate the behavior of this oscillatory component.
\begin{figure}
\begin{centering}
\includegraphics[scale=0.4]{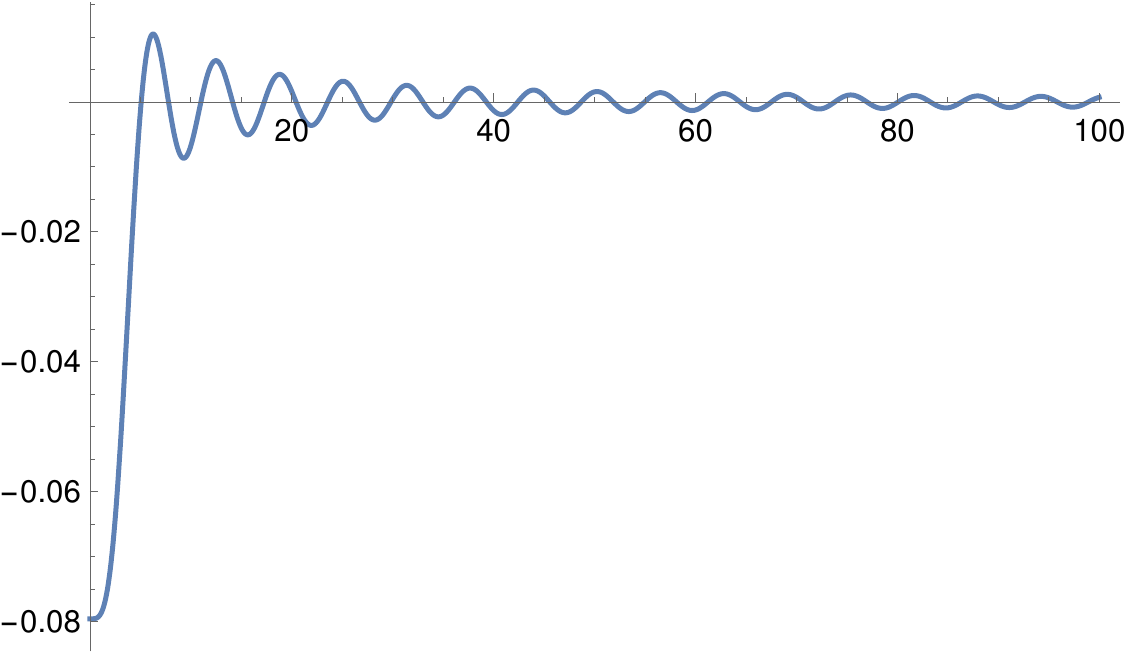}\includegraphics[scale=0.4]{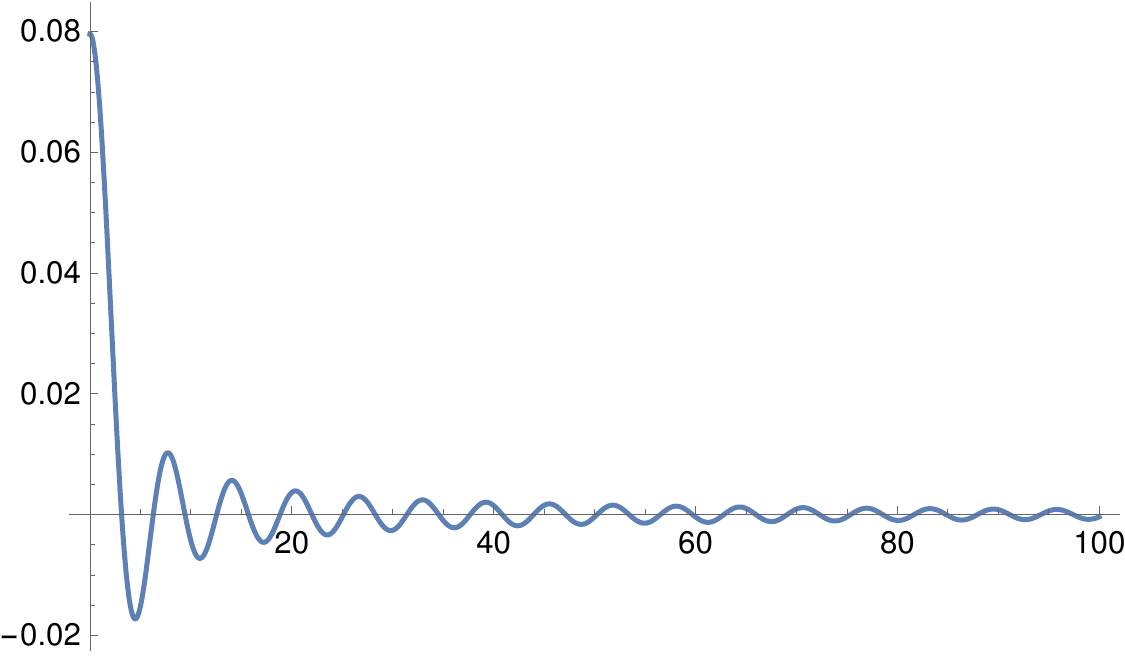}
\par\end{centering}
\begin{centering}
(a)~~~~~~~~~~~~~~~~~~~~~~~~~~~~~~~~~~~~~~~~~~~~~~~~~~~~~~~~~~~~~~~~~~~~~~~~~~~~~~~~~~~~~~~(b)
\par\end{centering}
\caption{\label{fig:Real-and-imaginary}The real part $\mathcal{R}e\left(g_{4,oscill}\right)\left(r,k\right)=\frac{1}{4\pi}\cos\left(kr\right)/r-g_{4}\left(r,k\right)$
(\ref{eq:difference q_n}) (a) and the imaginary part $\frac{1}{4\pi}\sin\left(kr\right)/r$
of the oscillatory component (b) for $k=1$.}
\end{figure}

\section{Spatial representations in $\mathbb{R}^{2}$}

To avoid confusion, we denote the inverse Fourier transform of (\ref{eq: non-oscillatory})
in dimension $d=2$ as $h_{n}\left(r\right)$. We have
\begin{lem}
The inverse Fourier transform of (\ref{eq: non-oscillatory}) in dimension
$d=2$ yields 
\begin{equation}
h_{n}\left(r,k\right)=\frac{1}{2\pi}\sum_{j=0}^{n-1}\frac{\left(kr\right)^{j}}{j!}K_{j}\left(kr\right),\label{eq:h_n in 2D}
\end{equation}
where $K$ is the modified Bessel function of the second kind. 
\end{lem}

\begin{proof}
Using (\ref{eq: non-oscillatory}), we obtain
\begin{eqnarray}
h_{n}\left(r,k\right) & = & \frac{1}{\left(2\pi\right)^{2}}\int_{\mathbb{R}^{2}}\widehat{g}_{n}\left(\left\Vert \mathbf{p}\right\Vert \right)e^{i\mathbf{p}\cdot\mathbf{r}}d\mathbf{p}\nonumber \\
 & = & \frac{1}{2\pi}\sum_{j=0}^{n-1}\int_{0}^{\infty}\frac{\left(2k^{2}\right)^{j}}{\left(p^{2}+k^{2}\right)^{j+1}}J_{0}\left(pr\right)pdp\label{eq:integral 2d}
\end{eqnarray}
The integral in (\ref{eq:integral 2d}) is available in \cite[Formula 6.565.4]{GRA-RYZ:2007}
leading to (\ref{eq:h_n in 2D}).

Next we consider the difference between the real part of $G$ in (\ref{eq:  FreeSpaceGF})
and $h_{n}$ in dimension $d=2$,
\begin{equation}
v_{n}\left(r,k\right)=\mathcal{R}e\left(h_{n,oscill}\right)\left(r,k\right)=-\frac{1}{4}Y_{0}\left(kr\right)-h_{n}\left(r,k\right),\label{eq:v_n (r,k)}
\end{equation}
and examine its behavior at $r=0$.
\end{proof}
\begin{lem}
The difference
\begin{equation}
\mathcal{R}e\left(h_{n,oscill}\right)\left(\mathbf{r},k\right)=-\frac{1}{4}Y_{0}\left(k\left\Vert \mathbf{r}\right\Vert \right)-h_{n}\left(\left\Vert \mathbf{r}\right\Vert ,k\right)\label{eq:h_n oscill}
\end{equation}
has continuous partial derivatives at $\left\Vert \mathbf{r}\right\Vert =0$
up to order $2n$. The Taylor expansion of $v_{n}\left(r,k\right)$
yields

\begin{eqnarray}
v_{1}\left(r,k\right) & = & \frac{-1+\gamma-\log2}{4\pi}\left(kr\right)^{2}+\mathcal{O}\left(\left(kr\right)^{2}\log\left(kr\right)\right)\nonumber \\
v_{2}\left(r,k\right) & = & -\frac{1}{2\pi}-\frac{1}{8\pi}\left(kr\right)^{2}+\frac{5-4\gamma+4\log2}{128\pi}\left(kr\right)^{4}+\mathcal{O}\left(\left(kr\right)^{4}\log\left(kr\right)\right)\nonumber \\
v_{3}\left(r,k\right) & = & -\frac{1}{\pi}+\frac{1}{64\pi}\left(kr\right)^{4}+\frac{-4+3\gamma-3\log2}{1728\pi}\left(kr\right)^{6}+\mathcal{O}\left(\left(kr\right)^{6}\log\left(kr\right)\right)\nonumber \\
v_{4}\left(r,k\right) & = & -\frac{5}{3\pi}+\frac{1}{12\pi}\left(kr\right)^{2}+\frac{1}{192\pi}\left(kr\right)^{4}-\frac{5}{6912\pi}\left(kr\right)^{6}\label{eq:2d formulas}\\
 & + & \frac{11-8\gamma+8\log2}{147456\pi}\left(kr\right)^{8}+\mathcal{O}\left(\left(kr\right)^{8}\log\left(kr\right)\right)\nonumber \\
v_{5}\left(r,k\right) & = & -\frac{8}{3\pi}+\frac{1}{6\pi}\left(kr\right)^{2}-\frac{1}{3456\pi}\left(kr\right)^{6}+\frac{1}{55296\pi}\left(kr\right)^{8}\nonumber \\
 & + & \frac{-169+120\gamma-120\log2}{110592000\pi}\left(kr\right)^{10}+\mathcal{O}\left(\left(kr\right)^{10}\log\left(kr\right)\right)\nonumber \\
 & \dots & ,\nonumber 
\end{eqnarray}
where $\gamma$ is Euler's constant.
\end{lem}

\begin{proof}
The function $v_{n}$ in (\ref{eq:v_n (r,k)}) is the inverse Fourier
transform of rotationally invariant function (\ref{eq:g hat oscill})
in dimension $d=2$, 
\[
\mathcal{R}e\left(h_{n,oscill}\right)\left(r,k\right)=\frac{1}{2\pi}\int_{0}^{\infty}\frac{p}{p^{2}-k^{2}}\frac{\left(2k^{2}\right)^{n}}{\left(p^{2}+k^{2}\right)^{n}}J_{0}\left(pr\right)pdp.
\]
We use the same argument to count the number of continuous derivatives
of $h_{n,oscill}\left(\mathbf{r},k\right)$ as in Lemma~\ref{lem:The-Taylor-expansion}.
Several examples of expansions of (\ref{eq:v_n (r,k)}) using (\ref{eq:h_n in 2D})
are presented in (\ref{eq:2d formulas}) .
\end{proof}

\section{Integral representations}

Integral representations via Gaussians of kernels of non-oscillatory
operators have been used as a starting point to obtain their accurate
multiresolution approximations via a linear combination of Gaussians,
see e.g. \cite{H-F-Y-G-B:2004,BEY-MOH:2005,B-C-F-H:2007,BE-MO-PE:2008,BEY-MON:2010,B-F-H-K-M:2012,H-B-B-C-F-F-G-etc:2016,A-H-S-T-B:2020}.
For any $\epsilon>0$, kernels are approximated with accuracy $\epsilon$
by a linear combination of Gaussians where the number of terms is
shown to be $\mathcal{O}\left(\left(\log\epsilon^{-1}\right)^{2}+\log\delta^{-1}\right)$,
where $\delta$ defines the interval of validity of the approximation,
e.g. $\delta<r<\delta^{-1}$ (see e.g. \cite{BEY-MON:2010}). This
estimate is somewhat conservative since the actual number of terms
appears to be $\mathcal{O}\left(\log\epsilon^{-1}+\log\delta^{-1}\right)$.
In what follows we construct integral representations via Gaussians
of the non-oscillatory components of the Green's function (\ref{eq:  FreeSpaceGF}).
\begin{lem}
The function $g_{n}\left(r,k\right)$ in dimension $d=3$ has an integral
representation
\begin{equation}
g_{n}\left(r,k\right)=\int_{-\infty}^{\infty}e^{-r^{2}e^{t}/4}w_{n}\left(k,t\right)dt\label{eq: integral repr g_n in 3D}
\end{equation}
where
\[
w_{n}\left(k,t\right)=\frac{1}{8\pi^{3/2}}e^{-k^{2}e^{-t}+\frac{1}{2}t}\left(\sum_{j=0}^{n-1}\frac{\left(2k^{2}e^{-t}\right)^{j}}{j!}\right).
\]
\end{lem}

\begin{proof}
Using the integral (see e.g. \cite[eq. 2]{BEY-MON:2010})
\[
\rho^{-j-1}=\frac{1}{j!}\int_{-\infty}^{\infty}e^{-\rho e^{t}+\left(j+1\right)t}dt,
\]
we obtain from (\ref{eq: explicit g_n(r) in 3D}) and (\ref{eq: non-oscillatory})
\begin{eqnarray}
g_{n}\left(r,k\right) & = & \frac{1}{\left(2\pi\right)^{3}}\sum_{j=0}^{n-1}\int_{\mathbb{R}^{3}}\frac{\left(2k^{2}\right)^{j}}{\left(\left\Vert \mathbf{p}\right\Vert ^{2}+k^{2}\right)^{j+1}}e^{i\mathbf{p}\cdot\mathbf{r}}d\mathbf{p}\nonumber \\
 &  & \frac{1}{\left(2\pi\right)^{3}}\int_{\mathbb{R}^{3}}\left(\sum_{j=0}^{n-1}\frac{\left(2k^{2}\right)^{j}}{j!}\int_{-\infty}^{\infty}e^{-\left(k^{2}+\left\Vert \mathbf{p}\right\Vert ^{2}\right)e^{t}+\left(j+1\right)t}dt\right)e^{i\mathbf{p}\cdot\mathbf{r}}d\mathbf{p}\nonumber \\
 & = & \frac{1}{\left(2\pi\right)^{3}}\int_{\mathbb{R}^{3}}\left(\int_{-\infty}^{\infty}e^{-e^{t}\left\Vert \mathbf{p}\right\Vert ^{2}}\left(\sum_{j=0}^{n-1}\frac{\left(2k^{2}e^{t}\right)^{j}}{j!}\right)e^{-k^{2}e^{t}+t}dt\right)e^{i\mathbf{p}\cdot\mathbf{r}}d\mathbf{p}\nonumber \\
 & = & \int_{-\infty}^{\infty}\left(\frac{1}{\left(2\pi\right)^{3}}\int_{\mathbb{R}^{3}}e^{-e^{-t}\left\Vert \mathbf{p}\right\Vert ^{2}}e^{i\mathbf{p}\cdot\mathbf{r}}d\mathbf{p}\right)\sum_{j=0}^{n-1}\frac{\left(2k^{2}e^{-t}\right)^{j}}{j!}e^{-k^{2}e^{-t}-t}dt,\label{eq:g_n in 3D}
\end{eqnarray}
where in the last integral we changed the order of integration and
replaced $t$ by $-t$ for convenience. Computing the inverse Fourier
transform of rotationally invariant function $e^{-e^{-t}\left\Vert \mathbf{p}\right\Vert ^{2}}$
in dimension $d=3$, we have 
\begin{eqnarray}
\frac{1}{\left(2\pi\right)^{3}}\int_{\mathbb{R}^{3}}e^{-e^{-t}\left\Vert \mathbf{p}\right\Vert ^{2}}e^{i\mathbf{r}\cdot\mathbf{p}}d\mathbf{p} & = & \frac{1}{2\pi^{2}r}\int_{0}^{\infty}e^{-e^{-t}p^{2}}\sin\left(pr\right)pdp\label{eq:intermediate}\\
 & = & \frac{1}{8\pi^{3/2}}e^{-r^{2}e^{t}/4}e^{\frac{3}{2}t}.\nonumber 
\end{eqnarray}
Substituting (\ref{eq:intermediate}) into (\ref{eq:g_n in 3D}),
we arrive at (\ref{eq: integral repr g_n in 3D}).
\end{proof}
Turning to dimension $d=2$, we have
\begin{lem}
The function $h_{n}$ in dimension $d=2$ has an integral representation
\begin{equation}
h_{n}\left(r,k\right)=\int_{-\infty}^{\infty}e^{-r^{2}e^{t}/4}\omega_{n}\left(k,t\right)dt,\label{eq:integral representation g_n in 2D}
\end{equation}
where
\[
\omega_{n}\left(k,t\right)=\frac{1}{4\pi}e^{-k^{2}e^{-t}}\left(\sum_{j=0}^{n-1}\frac{\left(2k^{2}e^{-t}\right)^{j}}{j!}\right).
\]
\end{lem}

\begin{proof}
Using integral representation of functions $K_{j}$ derived in \cite[eq. 36]{BEY-MON:2010},
we have

\[
\left(kr\right)^{j}K_{j}\left(kr\right)=2^{j-1}k^{2j}\int_{0}^{\infty}e^{-r^{2}s/4-k^{2}/s}s^{-j-1}ds,
\]
and, therefore,
\[
\frac{1}{2\pi}\sum_{j=0}^{n-1}\frac{\left(kr\right)^{j}}{j!}K_{j}\left(kr\right)=\frac{1}{4\pi}\int_{0}^{\infty}e^{-r^{2}s/4-k^{2}/s}\sum_{j=0}^{n-1}\frac{\left(2k^{2}\right)^{j}}{j!}s^{-j-1}ds.
\]
Changing variables $s=e^{t}$, we obtain (\ref{eq:integral representation g_n in 2D}). 
\end{proof}

\section{Discretization of integral representations}

Spatial integral representations of non-oscillatory components in
(\ref{eq: integral repr g_n in 3D}) and (\ref{eq:integral representation g_n in 2D})
lead to approximations of functions $g_{n}$ and $h_{n}$ by a linear
combination of Gaussians. In both cases, for $r>\delta$, the integrands
decay super exponentially for $t\to\pm\infty$, where the rate of
decay for $t\to-\infty$ is controlled by $k$. Heuristically, by
selecting a finite interval of integration so that the integrands
and their derivatives are negligible outside that interval, any user
selected accuracy $\epsilon$ can be achieved using the trapezoidal
rule. The resulting sum is a linear combination of Gaussians which
may be viewed as a multiresolution approximation of the non-oscillatory
component. 

As a result, we obtain a separated multiresolution approximation of
the kernel of the non-oscillatory component of the Helmholtz operator.
There are several approaches to apply this kernel rapidly in $\mathcal{O}\left(k^{d}\left(\log\epsilon^{-1}\right)^{2}\right)$
operations, for example via algorithms in \cite{BE-CH-PE:2008} or
via the Fast Gauss transform in \cite{GRE-STR:1991,GRE-SUN:1998}.
We refer to \cite{BE-KU-MO:2009,BE-KU-MO:2008} for the details of
algorithms for applying both the oscillatory and the non-oscillatory
(singular) components.

As an example, we describe an approximation of $g_{4}\left(r,k\right)$
where $k=100$. We discretize the integral in (\ref{eq: integral repr g_n in 3D})
as 
\begin{equation}
\widetilde{g}_{4}\left(r,k\right)=\Delta\sum_{m=20}^{200}e^{-r^{2}e^{m\Delta}/4}w_{4}\left(k,m\Delta\right),\,\,\,\Delta=\frac{1}{4},\,\,\,k=100,\,\,\,r=\left\Vert \mathbf{r}\right\Vert .\label{eq:lin_comb_gauss}
\end{equation}
 Near the singularity of $g_{4}\left(r,k\right)$, on the interval
$10^{-10}\le r\le10^{-3}$, we use the relative error
\begin{equation}
e_{0}\left(r\right)=\log_{10}\left|\frac{g_{4}\left(r,k\right)-\widetilde{g}_{4}\left(r,k\right)}{g_{4}\left(r,k\right)}\right|,\label{eq:rel_err}
\end{equation}
and on the interval $10^{-3}\le r\le40/k+2\log_{10}k/k$ (or $10^{-3}\le r\le44/100$,
see Remark~\ref{rem:We-want-to}) the absolute error,
\begin{equation}
e_{1}\left(r\right)=\log_{10}\left|g_{4}\left(r,k\right)-\widetilde{g}_{4}\left(r,k\right)\right|.\label{eq:abs_err}
\end{equation}
These approximation errors are illustrated in Figure~\ref{fig:errors}.
The terms of the linear combination of Gaussians in (\ref{eq:lin_comb_gauss})
can be applied to a function in parallel. Also we note that the Gaussians
with large exponents can be treated as approximations to a delta function
and such terms can be combined reducing the overall number of terms.

\begin{figure}
\begin{centering}
\includegraphics[scale=0.3]{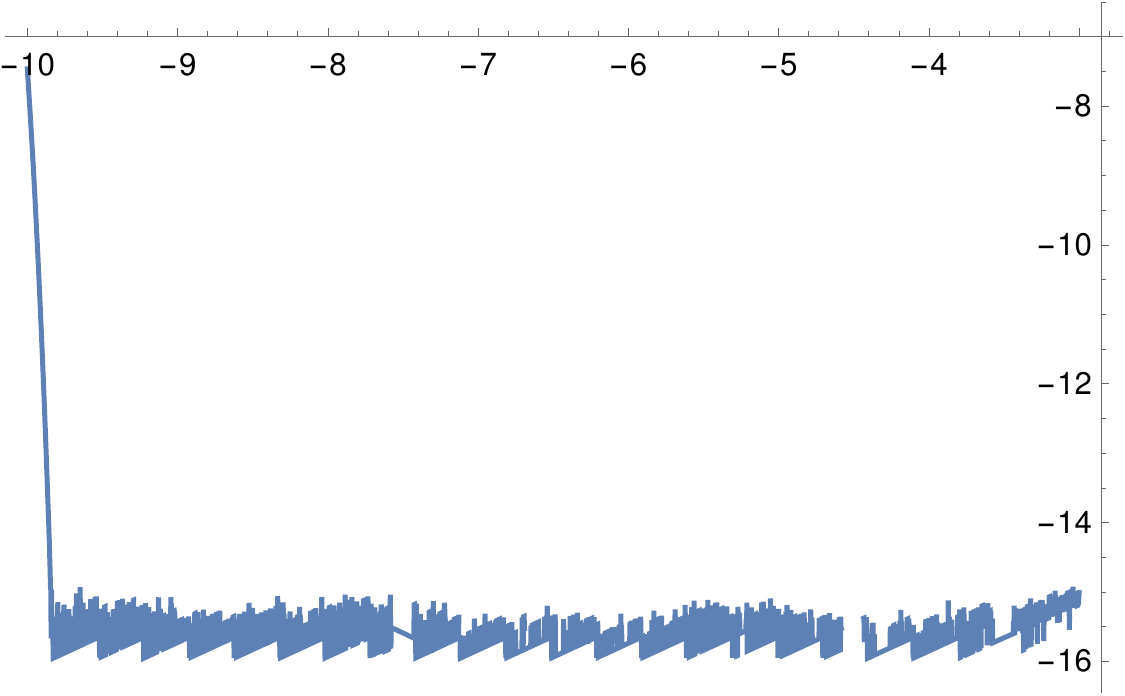}\includegraphics[scale=0.3]{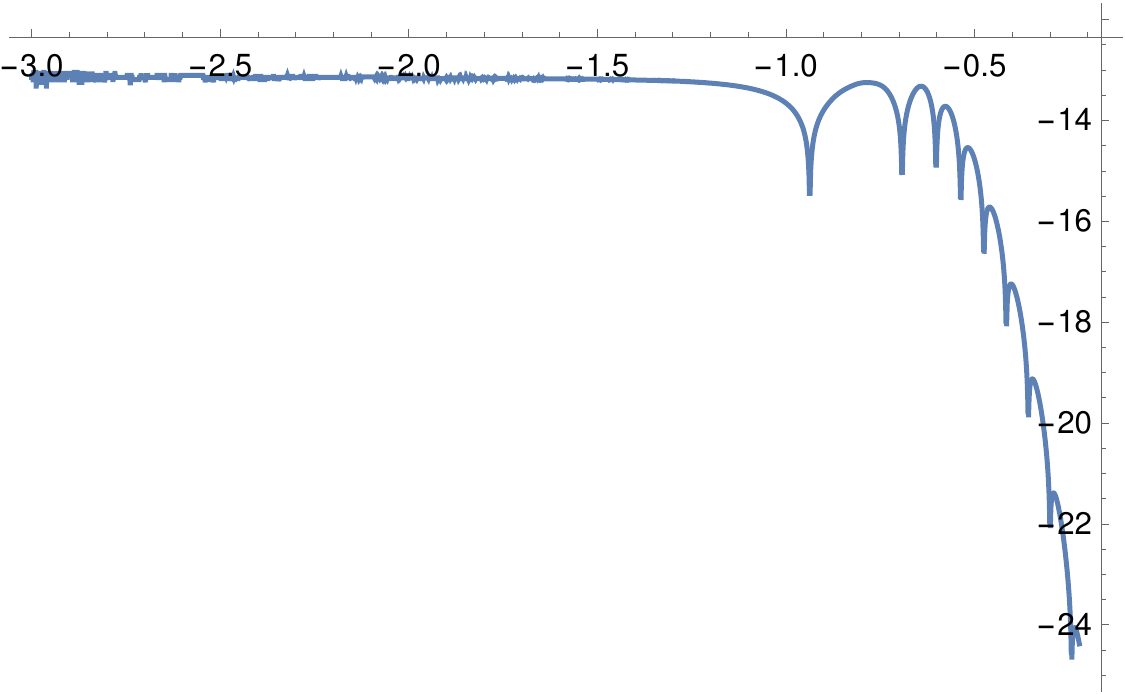}
\par\end{centering}
\caption{\label{fig:errors}The relative error $e_{0}\left(10^{x}\right)$
in (\ref{eq:rel_err}) for $-10\le x\le3$ and the absolute error
$e_{1}\left(10^{x}\right)$ in (\ref{eq:abs_err}) for $-3\le x\le\log_{10}\left(0.44\right)$. }

\end{figure}

\section{Conclusions}

The Helmholtz operator appears in many problems of mathematical physics
and is also part of more complicated Green's functions. In particular,
the components of the Dyadic Green's function in electromagnetics
are derivatives of the the Helmholtz Green's function. Therefore,
the splitting into the oscillatory and the non-oscillatory components
can be obtained for the Dyadic Green's function as well. We expect
several problems beyond the one described in this paper to be addressed
using our results.

\section{Acknowledgments}

The author would like to thank Brad Alpert (NIST) and Lucas Monz\'{o}n
(CU) for suggestions to improve the manuscript.

\bibliographystyle{plain}

\end{document}